\RequirePackage{fix-cm}
\documentclass[smallextended,natbib]{svjour3}
\smartqed
\usepackage{graphicx}
\usepackage{hyperref}
\hypersetup{
    colorlinks=true,
    linkcolor=blue,
    citecolor=blue,
    filecolor=magenta,      
    urlcolor=cyan,
    }

\usepackage{mathtools}
\usepackage{amsfonts,dsfont}
\usepackage{mathabx}
\usepackage{microtype}
\usepackage{sgamex}
\newtheorem{gamedefn}{Game}

\newcommand\reals{\mathbb{R}}

\newcommand\EXP{\mathbb E}
\newcommand\PR{\mathds {P}}

\newcommand{\IND}{\mathds{1}}
\newcommand\ALPHABET{\mathcal}

\newcommand*\Player[1]{\ensuremath{\mathsf{P}_{#1}}}

\newcommand{\teamset}{\ALPHABET{K}}
\newcommand{\numteams}{K}
\newcommand{\team}{k}
\newcommand\agentsteam[1]{N^{(#1)}}
\newcommand\agentsteamset[1]{\mathcal{N}^{(#1)}}
\newcommand\StSpteam[1]{\ALPHABET{\St}^{(#1)}}
\newcommand\AcSpteam[1]{\ALPHABET{\Ac}^{(#1)}}
\newcommand{\St}{S}
\newcommand{\Ac}{A}
\newcommand{\st}{s}
\newcommand{\ac}{a}
\newcommand\Stteam[1]{\St^{(#1)}}
\newcommand\Acteam[1]{\Ac^{(#1)}}
\newcommand\stteam[1]{\st^{(#1)}}

\newcommand{\mffunc}{\xi}
\newcommand{\MF}{Z}
\newcommand{\mf}{z}
\newcommand\MFteam[1]{\MF^{(#1)}}
\newcommand\mfteam[1]{\mf^{(#1)}}
\newcommand\MFteamSpace[1]{\ALPHABET \MF^{(#1)}}
\newcommand\MFSpace{\ALPHABET \MF^*}
\newcommand{\SMF}{\bar{\MF}}
\newcommand{\smf}{\bar{\mf}}

\newcommand\smfteam[1]{\smf^{(#1)}}

\newcommand\SMFSpace{\bar{\ALPHABET \MF^*}}
\newcommand\stprobteam[1]{P^{(#1)}}
\newcommand{\info}{I}
\newcommand{\pol}{\pi}

\newcommand\polteam[1]{\pol^{(#1)}}
\newcommand\optpolteam[1]{\pol^{*(#1)}}
\newcommand{\cost}{c}
\newcommand{\Cost}{C}
\newcommand\costteamagent[1]{\cost^{(#1)}}
\newcommand\costteam[1]{\Cost^{(#1)}}
\newcommand{\costfunction}{\ell}
\newcommand\costfunctionteam[1]{\costfunction^{(#1)}}
\newcommand{\Costfunction}{L}
\newcommand\Costfunctionteam[1]{\Costfunction^{(#1)}}

\newcommand{\scostfunction}{\bar{\ell}}
\newcommand\scostfunctionteam[1]{\scostfunction^{(#1)}}

\newcommand{\perf}{J}
\newcommand\perfteam[1]{\perf^{(#1)}}
\newcommand{\prescription}{\gamma}
\newcommand\prescriptionteam[1]{\prescription^{(#1)}}
\newcommand{\Prescription}{\Gamma}
\newcommand\Prescriptionteam[1]{\Prescription^{(#1)}}

\newcommand{\corlaw}{\psi}
\newcommand\corlawteam[1]{\corlaw^{(#1)}}

\newcommand{\scorlaw}{\bar{\psi}}
\newcommand\scorlawteam[1]{\scorlaw^{(#1)}}

\newcommand\optcorlawhist{\varphi^*}
\newcommand{\corlawhist}{\varphi}
\newcommand\corlawhistteam[1]{\corlawhist^{(#1)}}
\newcommand\optcorlawhistteam[1]{\corlawhist^{*(#1)}}

\newcommand{\MFGT}{\text{MFGT}}

\newcommand*\W{d_{\mathcal{W}}}
\begin{document}

\title{Mean-field games among teams\thanks{The work of AM was supported by Natural Sciences and Engineering Research Council of Canada, Discovery Grant RGPIN-2021-03511.}}
\author{%
Jayakumar Subramanian 
\and Akshat Kumar \and Aditya Mahajan}

\institute{J. Subramanian \at 
Media and Data Science Research Lab, Digital Experience Cloud, Adobe Inc., \\ Noida, Uttar Pradesh, India, \email{jasubram@adobe.com}
\and A. Kumar \at
School of Computing and Information Systems at the Singapore Management University, Singapore, \email{akshatkumar@smu.edu.sg}
\and A. Mahajan \at Department of Electrical and Computer
Engineering, \\ McGill University, Montreal, Canada, \email{aditya.mahajan@mcgill.ca}}

\date{Received: date / Accepted: date}
\maketitle

\begin{abstract}
In this paper, we present a model of a game among teams. Each team consists of a homogeneous population of agents. Agents within a team are cooperative while the teams compete with other teams. The dynamics and the costs are coupled through the empirical distribution (or the mean field) of the state of agents in each team. This mean-field is assumed to be observed by all agents. Agents have asymmetric information (also called a non-classical information structure). We propose a mean-field based refinement of the Team-Nash equilibrium of the game, which we call mean-field Markov perfect equilibrium (MF-MPE). We identify a dynamic programming decomposition to characterize MF-MPE. We then consider the case where each team has a large number of players and present a mean-field approximation which approximates the game among large-population teams as a game among infinite-population teams. We show that MF-MPE of the game among teams of infinite population is easier to compute and is an $\varepsilon$-approximate MF-MPE of the game among teams of finite population.
\end{abstract}


\maketitle
\section{Introduction}
Traditionally, agents in a multi-agent system are modeled either as cooperative agents who optimize a common system-wide objective or as strategic agents who optimize individual objectives. This difference in the behavior of agents leads to different conceptual difficulties and different solution concepts. As a result, the two settings are studied as separate sub-disciplines of decision theory: models with cooperative agents are studied under the heading of team theory~\citep{MarschakRadner_1972} and models with strategic agents are studied under the heading of game theory~\citep{VonMorgenstern_2007}. For the most part, these two research subdisciplines have evolved independently.

However, in recent years, many applications have emerged which may be considered as games among teams. Examples include:
multiple demand aggregators competing in the same electricity markets,
multiple ride-sharing companies competing in the same city,
multiple firms competing in the same industry with different levels of competitive advantages~\citep{WeintraubBenkardVanRoy_2008}, and
the DARPA Spectrum Sharing Challenge, where teams of multiple radio units compete with other such teams in the same geographic area~\citep{DARPA}.

In such applications, teams of agents compete with other teams of agents. These models are different from traditional team theory models because agents belonging to different teams have separate objectives and are, therefore, not cooperative. These models are also different from traditional game theory models because agents belonging to the same team can enter into pregame agreements; therefore, the notion of equilibrium in games among teams must account for the possibility of simultaneous and coordinated deviations by all agents belonging to the same team. Such an equilibrium is called a \emph{Team-Nash equilibrium}~\citep{tang2021dynamic}. 

Games among teams are also different from cooperative games~\citep{shapley1953value}. In a game among teams, the team structure (i.e, which agent belongs to which team) is pre-specified and unlike cooperative game theory, the process of team formation and how to distribute rewards among members of the team is not investigated. 

\begin{figure}[htb]
  \centering
\begin{game}{2}{2}[$I$]
& $L$ & $R$\\
$T$ &$3,3,1$ &$0,0,0$\\
$B$ &$0,0,0$ &$1,1,3$
\end{game}\qquad
\begin{game}{2}{2}[$II$]
& $L$ & $R$\\
$T$ &$0,0,0$ &$1,1,5$\\
$B$ &$2,2,2$ &$0,0,0$
\end{game}
\caption{A game between the team of \Player1 (who chooses the row) and \Player2 (who chooses the column) versus \Player3 (who chooses the matrix $I$ or $II$).}
\label{fig:ex1}
\end{figure}

To illustrate the difference between Nash equilibrium (NE) and Team-Nash equilibrium (TNE), we consider a static game among two teams shown in Fig.~\ref{fig:ex1}. If we view the above as multiplayer game with three players, then the game has four NE in pure strategies: $\{ (T,L,I), (B,R,I), (T,R,II), (B,L,II) \}$. Of these, only $\{ (T,L,I), (B,R,II) \}$ are TNE. $(B,R,I)$ is not TNE because the team of \Player1 and \Player2 can deviate to $(T,L)$ and obtain a higher payoff. Thus, the key difference between NE and TNE is that in TNE players belonging to the same team can deviate together. 

When agents in a team have symmetric information (as is the case in the example above), the game among teams can be reduced to a regular game by considering a single player with vector-valued actions. However, such a reduction is not possible when players belonging to the same team have \emph{asymmetric information}. The situation is even more complicated for multi-stage (or dynamic) games, where games among teams inherit all the conceptual challenges of dynamic games with asymmetric information. 

In many of the motivating applications described above, each team has a large number of agents. So, we investigate games among teams where each team has a large number of agents and call such systems \emph{mean-field games among teams} (or \MFGT\ for short).  In the last decade, various \emph{mean-field approximations} of teams and games with a large number of agents have been proposed in the literature. The general flavor of the results are as follows. 
For games, it is shown that 
an appropriate refinement of a Nash equilibrium of a game with a large number of players can be approximated by an equilibrium solution of a game with an infinite number of players~\cite{HuangCainesMalhame_2007,HuangCainesMalhame_2012,LasryLions_2007,WeintraubBenkardVanRoy_2008} (and follow up literature).
Similarly, for teams, it is shown that a globally optimal solution of a team with a large number of players can be approximated by an optimal solution of a team with infinite number of players~\cite{ArabneydiMahajan_2014, ArabneydiMahajan_2016,bauerle2023mean,elliott2013discrete}.  

Our main contribution is to show that a similar high-level idea works for games among teams. In particular:
\begin{itemize}
  \item We present a model of multi-stage games among teams where each team has a large number of agents. 
  \item We propose a mean-field based Markov perfect equilibrium (MF-MPE) for the game among teams and present a dynamic programming decomposition to compute MF-MPE.
  \item When each team has a large number of players, we approximate the system with a game among teams with infinite players and show that any MF-MPE of the infinite population game among teams is an $\varepsilon$-approximate MF-MPE of the large population game among teams, where we provide an upper bound on $\varepsilon$. Our approximation results are different from typical infinite-population approximations in mean-field games (MFG). In MFG, as the number of players becomes large, each player has negligible influence on the dynamics and payoffs of other agents. However, in the game among teams, the number of teams remains fixed, so a different approach is required to establish the approximation results.
\end{itemize}

The rest of the paper is organized as follows. In Sect.~\ref{sec:review}, we review the relevant literature. In Sect.~\ref{sec:model}, we present the system model and problem formulation for mean-field game among teams. In Sect.~\ref{sec:solution}, we present an equivalent game among coordinators for teams and show that any MPE (Markov perfect equilibrium) of the game among coordinators is a Team-Nash equilibrium of the original game, which we call MF-MPE (mean-field MPE). We also present a dynamic program to characterize MF-MPE. In Sec.~\ref{sec:mf}, we present a mean-field approximation of MF-MPE when each team has a large number of players and we conclude in Sect.~\ref{sec:conclusion}

\section{Literature Overview}\label{sec:review}
There has been some recent interest in modeling and analyzing games among teams. A dynamic game among teams with delayed intra-agent information sharing is considered in~\citet{tang2021dynamic}, where common-information based refinements for Team-Nash equilibrium are presented. The results of \citet{tang2021dynamic} consider teams with general heterogeneous agents and no simplifications due to homogeneity and large number of agents are consider. Such mean-field approximations for games among teams are considered in~\citet{pedram2019,yu2020teamwise,sanjari2022nash,guan2023zerosum} and we summarize their results below.

\cite{pedram2019} are motivated by transportation networks and focus on designing incentive mechanisms to mitigate congestion in routing games over graphs. Their main result is proposing a toll mechanism, establishing a mean-field limit of the resulting large population game among teams, and showing that the resulting mean-field equilibrium can be computed efficiently. Note that the dynamics and reward models considered in our setup are different from \citet{pedram2019}.

\citet{yu2020teamwise} are motivated large firms aiming to develop new products or technologies with a rank-based rewards, where each team member contributes to the jump intensity of a common Poisson process, and the reward received by each team depends on its ranking. Different settings are considered for determining the team size: (i)~by a central planner; (ii)~by a team manager; and (iii) by team members voting in a partnership. Their main result is to propose a relative performance criteria which enables an explicit computation of the equilibrium solution. The also establish that the equilibrium eliminates the free riding problem. Note that the dynamics and reweard models considered in our setup are different from \citet{yu2020teamwise}.

\citet{sanjari2022nash} consider a generalization of Witsenhausen's intrinsic model \citep{witsenhausen1975} in where there are multiple agents that act sequentially. Each agent acts only once. Agents belong to one of finite number of teams and all agents within a team are exchangeable. They establish that the setting with large number of agents in each team can be approximated by an infinite population mean-field limit. Moreover, there exists a Nash equilibrium for the infinite population limit which is symmetric (i.e., each agent in the team considers identical strategies) and independently randomized. Note that in contrast to \citet{sanjari2022nash} we consider dynamic games (where each agent acts more than once) and consider a refinement of Markov perfect equilibrium.
\citet{guan2023zerosum} zero-sum games between two teams is considered. Under some technical assumptions, it is established that the optimal strategies of each team can be computed via a common-information based dynamic programming decomposition. It is then established that such games among teams with large number of agents can be approximated by their mean-field limit. Note that in contrast to \citet{guan2023zerosum}, we consider general sum games.
\section{Model and problem formulation}\label{sec:model}
\subsection{Model of mean-field games among teams}
\subsubsection{State and action spaces}
Consider a multi-agent system with $\numteams$ teams of homogeneous agents. Team 
$\team \in \teamset \coloneqq \{1, \dots, \numteams\}$ consists of $\agentsteam{k}$ agents
denoted by the set $\agentsteamset{k}$,
with state space $\StSpteam{k}$ and action space $\AcSpteam{k}$. We assume
that $\StSpteam{k}$ and $\AcSpteam{k}$ are finite sets. At time $t$, the
state and action of a generic agent $i$ in the team $k$ are denoted by
$\St^i_t \in \StSpteam{k}$ 
and $\Ac^i_t \in \AcSpteam{k}$, respectively. Moreover, let
\begin{equation*}
  \Stteam{k}_t = (\St^i_t)_{i \in \agentsteamset{k}}
  \quad\text{and}\quad
  \Acteam{k}_t = (\Ac^i_t)_{i \in \agentsteamset{k}}
\end{equation*}
denote the states and actions of all agents in team $k$ and
\begin{equation*}
  \St_t = (\Stteam{k}_t)_{k \in \teamset} 
  \quad\text{and}\quad
  \Ac_t = (\Acteam{k}_t)_{k \in \teamset}
\end{equation*}
denote the global state and actions of the entire system.

Given $\stteam{k} = (\st^i)_{i \in \agentsteamset{k}}$, $\st^i \in
\StSpteam{k}$, we use
$\mffunc(\stteam{k})$ to denote the mean field (or empirical distribution) of $\stteam{k}$, i.e., 
\begin{equation*}
  \mffunc(\stteam{k}) = \frac{1}{\agentsteam{k}}\sum_{i \in
  \agentsteamset{k}}\delta_{\st^i},
\end{equation*}
where $\delta_{\st}$ is a Dirac delta measure centered at $\st$. 
We use $\MFteam{k}_t = \mffunc(\Stteam{k}_t)$ to denote the mean field of the team $k$, $\MFteamSpace{k}$ to denote the space of realizations of $\MFteam{k}_t$,
$\MF_t = (\MFteam{k}_t)_{k \in \teamset}$ to denote the mean field of the
entire system, and $\MFSpace$ to denote the space of realizations of
$\MF_t$. With a slight abuse of notation, we use $Z_t = \mffunc(\St_t)$ to denote the mean-field 
of the entire system corresponding to the global state $\St_t$. 

\subsubsection{System dynamics} 
We use $(\st_{1:T}, \ac_{1:T})$ to denote a realization of $(\St_{1:T}, \Ac_{1:T})$
and $\mfteam{k}_t = \mffunc(\stteam{k}_t)$ to denote the mean field at time $t$. We assume that the initial states
of all agents are independent, i.e., 
\[
  \PR(\St_1 = \st_1)  = 
  \smashoperator[l]{\prod_{k \in \teamset}}
  \smashoperator[r]{\prod_{i \in \agentsteamset{k}}}
  \PR(\St^i_1 = \st^i_1) 
   \eqqcolon 
  \smashoperator[l]{\prod_{k \in \teamset}}
  \smashoperator[r]{\prod_{i \in \agentsteamset{k}}}
   \stprobteam{k}_0(\st^i_1),
\]
where $\stprobteam{k}_0$ denotes the initial state distribution of agents in
team $k$. The global state of the system evolves independently across agents in a
controlled Markov manner, i.e., 
\begin{multline*}
  \PR( \St_{t+1} = \st_{t+1} \mid \St_{1:t} = \st_{1:t}, \Ac_{1:t} = \ac_{1:t}) 
  \\
 =\prod_{k \in \teamset}
   \prod_{i \in \agentsteamset{k}}\PR(\St^i_{t+1} = \st^i_{t+1} \mid \St_t =
   \st_t, \Ac_t = \ac_t).
\end{multline*}
Finally, all agents within a team are exchangeable. So the state evolution of a generic agent depends on 
the states and actions of other agents only through the mean-fields of the states of the teams, i.e., 
for agent $i$ in team $k$,
\begin{align*}
  \PR(\St^i_{t+1} = \st^i_{t+1}  \mid \St_t = \st_t, \Ac_t = \ac_t) 
  &= \PR(\St^i_{t+1} = \st^i_{t+1}  \mid \St^i_t=\st^i_t, \Ac^i_t = \ac^i_t, \MF_t = \mf_t) 
  \\
   &\eqqcolon \stprobteam{k}(\st^i_{t+1} \mid \st^i_t, \ac^i_t, \mf_t),
\end{align*}
where $\stprobteam{k}$ denotes the controlled transition matrix for team~$k$.

Combining all of the above, we have 
\begin{equation}
  \PR(\St_{t+1} = \st_{t+1}  \mid \St_{1:t} = \st_{1:t}, \Ac_{1:t} = \ac_{1:t})   
   = \prod_{k \in \teamset}\prod_{i \in \agentsteamset{k}}\stprobteam{k}(\st^i_{t+1} \mid \st^i_t, \ac^i_t, \mf_t).
  \label{eq:dynamics}
\end{equation}

\subsubsection{Cost function}
There is a cost function: $\costteamagent{k}_t: \StSpteam{k} \times \AcSpteam{k} \times \MFSpace \to \reals$ associated with each agent in team $k$. The per-step cost incurred by team $k$ is the average of the cost associated with all agents in the team, i.e, 
\begin{equation} \label{eq:ind-to-team-cost}
  \costteam{k}_t = \frac{1}{\agentsteam{k}}\sum_{i \in
  \agentsteamset{k}}\costteamagent{k}_t(\St^i_t, \Ac^i_t, \MF_t).
\end{equation}

\subsubsection{Information structure and control laws}
We assume that the system has mean-field sharing
information-structure~\citep{ArabneydiMahajan_2014}, i.e., each agent has access to its local state $\St^i_t$ and the history of mean-field $\MF_{1:t}$ of all teams. Thus, the information
available to agent $i$ is given by:
\begin{equation}
  \info^i_t = \{\St^i_t, \MF_{1:t}\}.
\end{equation}
In addition, we assume that all agents in team $k$ use identical%
\footnote{Restricting attention to identical strategies for all
  agents in a team may result in a loss of optimality for that team. However, as
  argued in~\citet{ArabneydiMahajan_2014}, such a restriction is often
  justifiable due to other considerations such as simplicity, robustness and
fairness.}
(stochastic) control laws:
\(
  \polteam{k}_t \colon \StSpteam{k} \times {\MFSpace}^t 
   \to \Delta(\AcSpteam{k})
\)
to choose the action at time $t$, i.e., 
\begin{equation}
  \Ac^i_t \sim \polteam{k}_t(\St^i_t, \MF_{1:t}),
\end{equation}
where each agent randomizes independently.
Let $\polteam{k} \coloneqq (\polteam{k}_1, \dots, \polteam{k}_T)$ denote the
strategy of team~$k$. We use $\polteam{-k}$ to denote the strategy of all teams
other than~$k$. 
Given a strategy $\pol = (\polteam{k}, \polteam{-k})$ for all teams, the expected
total cost incurred by team $k$ is given by the following:
\begin{equation}
  \perfteam{k}(\polteam{k}, \polteam{-k}) = \EXP^{(\polteam{k},
\polteam{-k})}\biggl[ \sum_{t=1}^T \costteam{k}_t \biggr].
\end{equation}

It is assumed that the system dynamics $(\stprobteam{k}_0)_{k \in \teamset}$, $(\stprobteam{k})_{k \in \teamset}$, 
the cost functions $(\costteamagent{k})_{k \in \teamset}$ and the information structure are
common knowledge for all agents. Each team is interested in minimizing the total
expected cost incurred over a finite horizon. 
Following~\citet{tang2021dynamic}, we say that a strategy profile $\pol^* = (\pol^{*(k)})_{k \in \teamset}$ is a \textbf{Team-Nash equilibrium} if for all teams $k \in \teamset$ and all other strategy profiles $\polteam{k}$ for team $k$, we have:
\begin{equation}
    \perfteam{k}(\optpolteam{k}, \optpolteam{-k}) \le
    \perfteam{k}(\polteam{k}, \optpolteam{-k}).
  \end{equation}
In the sequel, we refer to the model defined above as mean-field game among teams (\MFGT). 
We are interested in the following:
\begin{gamedefn}\label{game:original}
  Identify a Team-Nash equilibrium of the mean-field game among teams (\MFGT) model defined above. 
\end{gamedefn}

\subsection{Salient features of the model}
Some salient features of the \MFGT\ model are as follows:
\subsubsection{Dynamic Game with asymmetric information}
\MFGT\ is a dynamic game (also called stochastic game), where there is a state space model which describes the evolution of the state of the environment. Agents have imperfect and asymmetric information about the state of the environment. 
\subsubsection{All agents in a team are homogeneous}
We have assumed that all agents in a team have homogeneous dynamics and cost functions. This assumption is made only for ease of exposition. It is conceptually straightforward to extend the discussion to models with heterogeneous agents where the agents have multiple types. In fact, such a model can be converted into a model with homogeneous agents by augmenting the state space and considering the type of each agent to be a (static) component of its state.
\subsubsection{All agents in a team play identical strategies}
In a general Team-Nash equilibrium~\citep{tang2021dynamic}, all agents in a team are allowed to deviate together and in a correlated manner. However, we have imposed an additional assumption that all agents in a team play identical strategies. Under this assumption, when agents in a team consider deviations, they only consider deviations in which all agents of that team are playing identical (randomized) strategies.
\subsubsection{Agents in a team randomize independently}
We have assumed that all agents randomize independently. In principle, at each time, agents in the same team could randomize in a correlated manner, but we do not consider that setup in this paper. Correlated randomizations can be obtained by augmenting the state space of all agents in a team with a common correlating random variable, which is independent over time. 

\section{Mean-field based Markov perfect equilibrium (MF-MPE)} \label{sec:solution}

\subsection{Road map of the solution approach}

Our main conceptual idea is as follows. First, we start by an alternative, but equivalent representation of the mean-field in terms of the state counts. Borrowing the idea of prescriptions (i.e., partially evaluated strategies) from decentralized stochastic control~\citep{NayyarMahajanTeneketzis_2013}, we show that state-counts (and, therefore, the mean-field) is a controlled Markov process controlled by prescriptions. Using these results, we show that for any team~$k$ and any arbitrary but fixed strategy $\pi^{-k}$ of all teams other than~$k$, the mean-field $\{\MF_t\}_{t \ge 1}$ is an \emph{information state} for players of team~$k$ in the sense of \citet{aisjmlr} (see Proposition~\ref{prop:IS}). Therefore, we can follow the idea of game between virtual players introduced by \citet{Nayyar:game}, to propose a game between $K$ virtual players, where virtual player $k \in \teamset$ observes the mean-field $\MF_t$ and chooses the prescription for all agents in team~$k$. We call this game, Game~\ref{game:common}, and show that Game~\ref{game:common} is equivalent to Game~\ref{game:original}. In particular, any Nash equilibrium of Game~\ref{game:common} is a Team-Nash equilibrium of Game~\ref{game:original} (see Theorem~\ref{thm:equiv}). 

Game~\ref{game:common} is a dynamic game with perfect information. Following~\cite{maskin1988theory,maskin1988theory2}, we can identify a dynamic program to characterize the Markov perfect equilibria (i.e., a Nash equilibria where all players play Markov strategies) of Game~\ref{game:common} (see Theorem~\ref{thm:MPE}). By Theorem~\ref{thm:equiv}, any MPE identified by the dynamic program of Theorem~\ref{thm:MPE} is a Team-Nash equilibrium of Game~\ref{game:original}. We call such Team-Nash equilibria as MF-MPE (mean-field Markov perfect equilibria).

\subsection{A count-based representation of the mean-field}

We start by an alternative, but equivalent, representation of the mean-field in terms of state counts.
Count-based representation has been explored earlier in collective decentralized POMDPs~\citep{NguyenKL17} for systems with a single team. We generalize these ideas to mean-field systems with multiple teams. We start by defining different types of counts: 
\begin{itemize}
  \item \emph{State counts}, denoted by $M^{(k)}_t$, which count the number of agents of team $k$ in each state and is given by
    \begin{equation*}
      M^{(k)}_t(\st) = \smashoperator[l]{\sum_{i \in \agentsteam{k}}}
      \IND\{\St^i_t = \st\},
     \quad \forall \st \in \StSpteam{k}.
    \end{equation*}
  \item \emph{State-action counts}, denoted by $\widebar{M}^{(k)}_t$, count the number of agents of team $k$ in each state-action pair and is given by
  \[
      \widebar{M}^{(k)}_t(\st, \ac) = 
      \smash{\smashoperator[l]{\sum_{i \in \agentsteam{k}}}}
      \IND\{\St^i_t = \st,{} \Ac^i_t = \ac\},
      \quad
      \forall \st, \ac \in \StSpteam{k} \times \AcSpteam{k}.
 \]
  \item \emph{State-action-next state counts}, denoted by $\widehat{M}^{(k)}_t$, count the number of agents of team $k$ in each state-action-next-state tuples, and is given by:
  \[
      \widehat{M}^{(k)}_t(\st, \ac, \st') =
      \smash{\smashoperator[l]{\sum_{i \in \agentsteam{k}}}}
      \IND\{\St^i_t = \st, \Ac^i_t = \ac, \St^i_{t+1} = \st'\}, 
      \quad
       \forall \st, \ac, \st' \in \StSpteam{k} \times \AcSpteam{k} \times \StSpteam{k}.
  \]
\end{itemize}
Similar to mean-field, we use 
$M_t = (M^{(k)}_t)_{k \in \teamset}$, 
$\widebar M_t = (\widebar M^{(k)}_t)_{k \in \teamset}$, 
and
$\widehat M_t = (\widehat M^{(k)}_t)_{k \in \teamset}$
to denote the counts for the entire system.

Note that the state counts are equivalent to the empirical mean field, i.e., $M^{(k)}_t = \agentsteam{k} \MFteam{k}_t$, or equivalently, $\MFteam{k}_t = M^{(k)}_t/\agentsteam{k}$. We use $Z_t = \mu(M_t)$ to denote the vector of mean-fields equivalent to the vector $M_t$ of counts.

The main advantage of a count-based representation is that it captures the inherent symmetry in the model due to the homogeneity of agents. For example, the state dynamics~\eqref{eq:dynamics} may be written as
\begin{multline}
  \PR(\St_{t+1} = \st_{t+1}  \mid \St_{1:t} = \st_{1:t}, \Ac_{1:t} = \ac_{1:t})   \\
  = \prod_{k \in \teamset} 
  \smashoperator[r]{\prod_{\substack{(\st^i,\ac^i,\st^i_{+}) \in \\
  \StSpteam{k} \times \AcSpteam{k} \times \StSpteam{k}}}}
\stprobteam{k}(\st^i_{+} \mid \st, \ac^i, \mffunc(\st_t))^{\widehat{m}^{(k)}_t(\st^i, \ac^i, \st^i_{+})},
  \label{eq:dynamics-counts}
\end{multline}
where $\widehat{m}^{(k)}_t$ denotes the realization of $\widehat M^{(k)}_t$ along the sample path $(\st^i_{1:t+1}, \ac^i_{1:t})$. Furthermore, the average cost $\costteam{k}$ (which is given by~\eqref{eq:ind-to-team-cost}) may be written as
\begin{equation}\label{eq:cost-counts}
    \costteam{k}_t = 
    \smashoperator[l]{\sum_{\substack{\st^i \in \StSpteam{k} \\ \ac^i \in \AcSpteam{k}}}}
    \costteamagent{k}(\st^i, \ac^i, \mffunc(\st_t)) \bar{m}^{(k)}_t(\st^i,\ac^i),
\end{equation}
where $\bar{m}^{(k)}_t$ denotes the realization of $\widebar{M}^{(k)}_t$ along the sample path $(s^i_{1:t}, \ac^i_{1:t})$. 
Notice that computing the right-hand-side expressions in~\eqref{eq:dynamics-counts} and~\eqref{eq:cost-counts} require only aggregate information, namely counts
$\widehat m^k_t$ and $\widebar m^k_t$ for each team $k$. We will exploit such symmetries to present a simpler and equivalent representation of Game~\ref{game:original}.  

\subsection{Dynamics of the counts}
Given any strategy $\pol = (\polteam{1}, \dots, \polteam{\numteams})$ for the
system and any realization $\mf_{1:T}$ of the mean field $\MF_{1:T}$ (or equivalently, for any realization 
$m_t$ of the state counts $M_t$ and $z_t = \mu(m_t)$), we
define the following partial functions, which we call \emph{prescriptions} following the terminology of~\citet{NayyarMahajanTeneketzis_2013}:
\begin{equation} \label{eq:prescription}
  \prescriptionteam{k}_t = \polteam{k}_t(\cdot, \mf_{1:t}), \quad \forall k \in \teamset.
\end{equation}

When the realization $\mf_{1:t}$ is given, $\prescriptionteam{k}_t$ is a function
from $\StSpteam{k}$ to $\Delta(\AcSpteam{k})$. When $\MF_{1:t}$ is a random variable,
$\polteam{k}(\cdot, \MF_{1:t})$ is a random function from $\StSpteam{k}$ to
$\Delta(\AcSpteam{k})$ and we denote this random function by $\Prescriptionteam{k}_t$.
We use $\prescription_t$ to denote $(\prescriptionteam{1}_t, \dots,
\prescriptionteam{\numteams}_t)$ and $\Prescription_t$ to denote $(\Prescriptionteam{1}_t, \dots \Prescriptionteam{\numteams}_t)$.

Now we describe the dynamics of the state counts given the current state count $m_t$ and prescription $\prescriptionteam{k}_t)$.
\subsubsection{From state-counts to state-action counts} Let $m^{(k)}_t$ and $\widebar{m}^{(k)}_t$ be consistent values of state counts and state-action counts, i.e., 
\[
m^{(k)}_t(\st) = \sum_{\ac \in \AcSpteam{k}}\widebar{m}^{(k)}_t(\st, \ac),
\quad \forall \st \in \StSpteam{k}.
\]
Then, from a basic combinatorial counting argument, we get 
\begin{multline}
\PR(\bar{M}_{t}^{(k)} = \widebar{m}_t^{(k)} \mid M_t = m_t, \Prescriptionteam{k}_t 
= \prescriptionteam{k}_t)  \\ 
= \smashoperator[l]{\prod_{\st \in \StSpteam{k}}}
\Bigg[ 
  \frac{m_t^{(k)}(\st)!}{ \prod_{\ac \in \AcSpteam{k}} \bar{m}^{(k)}_t(\st, \ac)} 
\smashoperator[r]{\prod_{\ac \in \AcSpteam{k}}} \prescriptionteam{k}_t(\ac | \st)^{\bar{m}^{(k)}_t(\st, \ac)}\Bigg],
\label{eq:m-to-bar-m}
\end{multline}
which is a product of multinomial distributions.

\subsubsection{From state-action counts to state-action-state counts}
Let $m^{(k)}_t$, $\widebar{m}^{(k)}_t$ and $\widehat{m}^{(k)}_t$ be consistent values of state counts, state-action counts and state-action-next state counts, i.e., 
\[
    \widebar{m}^{k}_t(\st, \ac) = \sum_{\st' \in \StSpteam{k}} \widehat{m}^{(k)}_t(\st, \ac, \st'), \quad \forall \st, \ac \in \StSpteam{k} \times \Acteam{k}.
\]
Let $\mf_t$ be the mean-field corresponding to $(m^{(1)}_t, \dots, m^{(K)}_t)$.
Then, from a basic combinatorial counting argument, we get 
\begin{multline}
\PR(\widehat{M}^{(k)}_t =\widehat{m}^{(k)}_t \mid \widebar{M}^{(k)}_t = \widebar{m}^{(k)}_t, M_t = m_t) \\
= \smash{\prod_{\substack{\st \in \StSpteam{k} \\ \ac \in \AcSpteam{k}}}}
\Bigg[ \frac{\widebar{m}^{(k)}_t(\st, \ac)!}{\prod_{\st'} \widehat{m}^{(k)}_t(\st, \ac, \st')!} 
 \prod_{\st' \in \StSpteam{k}} \stprobteam{k}\big(\st' \mid \st, \ac, \mf_t)^{\widehat{m}^{(k)}_t(\st, \ac, \st')}\Bigg],
\label{eq:bar-m-to-hat-m}
\end{multline}
which is also a product of multinomial distributions.

\subsubsection{From state-action-state counts to updated state counts}
Let $\widehat m^{(k)}_t$ and $m^{(k)}_{t+1}$ be consistent values of state-action-state and state counts, i.e., 
\begin{equation}
m^{(k)}_{t+1}(\st') = 
\smashoperator[l]{\sum_{\st \in \StSpteam{k}}} 
\smashoperator[r]{\sum_{\ac \in \AcSpteam{k}}}
\widehat{m}^{(k)}_t(\st, \ac, \st'),
\quad \forall \st' \in \StSpteam{k}.
\label{eq:hat-m-to-m}
\end{equation}
Thus, if we ``marginalize'' the sampled state-action-next state count $\widehat{m}^{(k)}_t$, we will obtain the state count $m^{(k)}_{t+1}$. 

\subsection{A game between virtual players}

We start by establishing that the mean-field $Z_t$ is an information state (in the sense of \citet{aisjmlr}) for every team.
\begin{proposition}\label{prop:IS}
For any strategy $\pol = (\polteam{1}, \dots, \polteam{K})$ and any time $t \in \{1,\dots, T\}$, the mean field $\{\MF_t\}_{t \ge 1}$ is an information state for every team~$k$, i.e., the following properties hold for any realization $(\mf_{1:t}, \prescription_{1:t})$ of $(\MF_{1:t}, \Prescription_{1:t})$:
\begin{enumerate}
    \item[(P1)] \textbf{Sufficient for performance evaluation:}
\begin{align}
\EXP[\costteam{k}_t  \mid \MF_{1:t} = \mf_{1:t}, \Prescription_{1:t} = \prescription_{1:t}] 
&= \EXP[\costteam{k}_t \mid \MF_t = \mf_t, \Prescriptionteam{k}_t = \prescriptionteam{k}_t] 
\notag \\
&\eqqcolon \costfunctionteam{k}_t(\mf_t, \prescriptionteam{k}_t) \label{eq:cost_team}
\end{align}
    
    \item[(P2)] \textbf{Sufficient for predicting itself:}
    for any $\mf = (\mfteam{1}, \dots \mfteam{\numteams}) \in \MFSpace$, we have
  \begin{align}\label{eq:mf-evolve}
    \PR&(\MF_{t+1} = \mf \mid \MF_{1:t} = \mf_{1:t}, \Prescription_{1:t} = \prescription_{1:t}) \notag \\
    &= 
    \prod_{k \in \teamset} 
    \PR(\MFteam{k}_{t+1} = \MFteam{k} \mid \MF_{t} = \mf_{t}, \Prescriptionteam{k}_{t} = \prescriptionteam{k}_{t}) \notag \\
    &\eqqcolon
    \prod_{k \in \teamset} 
    Q^{(k)}(\mfteam{k} \mid \mf_t, \prescriptionteam{k}_t), \notag \\
    &\eqqcolon
    Q(z \mid \mf_t, \prescription_t),
  \end{align}
  where $Q^{(k)}(\mfteam{k} \mid \mf_t, \prescriptionteam{k}_t)$ may be computed by
    combining \eqref{eq:m-to-bar-m}--\eqref{eq:hat-m-to-m}.
\end{enumerate}
\end{proposition}
\begin{proof}
    Property (P1) follows from~\eqref{eq:cost-counts} and~\eqref{eq:m-to-bar-m}. Property (P2) follows from~\eqref{eq:m-to-bar-m}--\eqref{eq:hat-m-to-m} and the fact that there is a one-to-one relationship between the mean field $\MF_t$ and the state counts $M_t$.
\end{proof}

Following~\citet{Nayyar:game}, we consider a stochastic game between $\numteams$
virtual players. At time $t$, the state is $\MF_t = (\MFteam{1}_t, \dots,
\MFteam{\numteams}_t) \in \MFSpace$; virtual player $k \in \teamset$ observes
$\MF_t$, chooses the prescription $\prescriptionteam{k}_t: \StSpteam{k} \to \Delta(\AcSpteam{k})$,
and incurs a per-step cost $\costfunctionteam{k}(\MF_t,
\prescriptionteam{k}_t)$ given by~\eqref{eq:cost_team}. 
The initial state $\MF_1$ has a probability mass function given by:
\begin{align}
\PR(\MF_1 = \mf) &= \prod_{k \in \teamset}\PR(\MFteam{k}_1 = \mfteam{k}) \notag \\
&=\prod_{k \in \teamset}\sum_{\stteam{k} \in (\StSpteam{k})^{\agentsteam{k}}}\prod_{i \in \agentsteamset{k}}P_0^{(k)}(\st^i_1).
\end{align}
The state $\MF_t$
evolves in a controlled Markov manner according to~\eqref{eq:mf-evolve}. 

The information available to the virtual player at time~$t$ is the history of
mean-fields $\MF_{1:t}$. Virtual player~$k$ selects the prescription
according to a strategy $\corlawhistteam{k}$, i.e.,
\[
  \Prescriptionteam{k}_t = \corlawhistteam{k}(\MF_{1:t}).
\]
Let $\corlawhistteam{k} = (\corlawhistteam{k}_1, \dots, \corlawhistteam{k}_T)$ denote the strategy of virtual player $k$. Then, the total cost incurred by virtual player~$k$ is given by:
\begin{equation}
\Costfunctionteam{k}(\corlawhistteam{k}, \corlawhistteam{-k}) = \EXP\Bigl[\sum_{t=1}^T \costfunctionteam{k}_t(\MF_t, \Prescriptionteam{k})\Bigr].
\end{equation}
We are interested in the following:
\begin{gamedefn}\label{game:common}
  Given the system model described above, identify a Nash equilibrium strategy $\optcorlawhist = (\optcorlawhistteam{k})_{k \in \teamset}$, i.e.,
  $\optcorlawhistteam{k}_t \colon \MF_t \mapsto \Prescriptionteam{k}$, i.e., 
  for any other strategy $\corlawhist = (\corlawhistteam{k})_{k \in \teamset}$, we
  have
  \begin{equation} \label{eq:L}
    \Costfunctionteam{k}(\optcorlawhistteam{k}, \optcorlawhistteam{-k}) \le
    \Costfunctionteam{k}(\corlawhistteam{k}, \optcorlawhistteam{-k}),
    \quad \forall k \in \teamset.
  \end{equation}
\end{gamedefn}

\subsection{Relationship between Games~\ref{game:original}
and~\ref{game:common}}
We have the following result that connects the solutions of Game~\ref{game:original} and Game~\ref{game:common}.
\begin{theorem} \label{thm:equiv}
  Let $\pol = (\polteam{1}, \dots, \polteam{\numteams})$ be a Team-Nash equilibrium of
  Game~\ref{game:original}. Define a strategy $\corlawhist = (\corlawhistteam{1},
  \dots, \corlawhistteam{\numteams})$ for Game~\ref{game:common} as follows: for any
  $\mf_{1:t} \in {\MFSpace}^t$:
  \begin{equation} \label{eq:NE-forward}
    \corlawhistteam{k}_t(\mf_{1:t}) = \polteam{k}_t(\cdot, \mf_{1:t}).
  \end{equation}
  Then $\corlawhist$ is a Nash equilibrium for Game~\ref{game:common}.

  Conversely, let $\corlawhist = (\corlawhistteam{1}, \dots, \corlawhistteam{\numteams})$
  by any Nash equilibrium for Game~\ref{game:common}. Define a strategy $\pol = (\polteam{1},
  \dots, \polteam{k})$ for Game~\ref{game:original} as follows: for any $\st
  \in \StSpteam{k}$ and $\mf \in \MFSpace$,
  \begin{equation} \label{eq:NE-reverse}
    \polteam{k}_t(\st, \mf_{1:t}) = \corlawhistteam{k}_t(\mf_{1:t})(\st).
  \end{equation}
  Then $\pol$ is a Team-Nash equilibrium of Game~\ref{game:original}.
\end{theorem}
See Appendix~\ref{app:proof1}
\subsection{Markov perfect equilibrium for Game~\ref{game:common}}

Game~\ref{game:common} among virtual players is a game with perfect
information since all players choose prescriptions based on the history
$\MF_{1:t}$ of mean-field which is common knowledge between the players.
Proposition~\ref{prop:IS} implies that we can view the current
mean field $\MF_t$ as the ``state'' of the system.
Following~\citet{maskin1988theory,maskin1988theory2}, we restrict our attention to the
Markov perfect equilibrium for Game~\ref{game:common}, which can be thought of
as a subgame perfect equilibrium of Game~\ref{game:common} where all virtual
players are playing Markov strategies which map current state to
prescriptions. Such a Markov perfect equilibrium can be obtained using dynamic
programming as follows~\citep{maskin1988theory,maskin1988theory2}. 
\begin{theorem} \label{thm:MPE}
  Consider a strategy profile $\corlaw = (\corlawteam{k})_{k \in \teamset}$,
  where each virtual player is playing a Markov strategy, i.e., 
  $\corlawteam{k}_t \colon \MF_t \mapsto \Prescriptionteam{k}_t$.

  A necessary and sufficient condition for $\corlaw$ to be a Markov perfect
  equilibrium for Game~\ref{game:common} is that it satisfy the following
  conditions:
  \begin{enumerate}
    \item For each possible realization $\mf_T$ of $\MF_T$, define the value
      function for virtual player~$k$:
      \begin{equation} \label{eq:DP:T}
        V^{(k)}_T(\mf_T) = \min_{\prescriptionteam{k}_T} 
        \costfunctionteam{k}_T(\mf_T, \prescriptionteam{k}_T).
      \end{equation}
      Then, $\corlawteam{k}_T(\mf_T)$ must be a minimizing
      $\prescriptionteam{k}_T$ in~\eqref{eq:DP:T}.

    \item For $t \in \{T-1, \dots, 1\}$ and for each possible realization
      $\mf_t$ of $\MF_t$, recursively define the value function for virtual
      player~$k$:
      \begin{align} 
        V^{(k)}_t&(\mf_t) = 
        \smash{\smashoperator[l]{\sum_{i \in \agentsteam{k}}}}
        \EXP\big[ \costfunctionteam{k}_t(\mf_t, \prescriptionteam{k}_t) 
        + V^{(k)}_{t+1}(\MF_{t+1}) \bigm| \mathcal{F}^{(k)} \big]
      \label{eq:DP:t}
      \end{align}
      where the event $\mathcal{F}^{(k)} = \{\MF_t = \mf_t$,
        $\Prescriptionteam{k}_t = \prescriptionteam{k}_t$,
      $\Prescriptionteam{-k}_t = \corlawteam{-k}_t(\mf_t) \}$ and 
      the expectation is with respect to the
      distribution~\eqref{eq:mf-evolve}. Then, $\corlawteam{k}_t(\mf_t)$ must
      be a minimizing $\prescriptionteam{k}_t$ in~\eqref{eq:DP:t}.
  \end{enumerate}
\end{theorem}

\begin{remark}
Theorem~\ref{thm:MPE} states that the Markov perfect equilibrium for the
virtual players can be obtained by dynamic programming. Let $\corlaw$ be such
a Markov perfect equilibrium. Let $\pol$ be the policy obtained
by~\eqref{eq:NE-reverse}. Then, by Theorem~\ref{thm:equiv}, $\pol$ is a Team-Nash
equilibrium of Game~\ref{game:original}, which we call \emph{mean-field
based Markov perfect equilibrium} (MF-MPE).
\end{remark}

\begin{remark}
  In general, solving the dynamic program of Theorem~\ref{thm:MPE} suffers from the curse of dimensionality. The space $\MFteamSpace{k}$ has at most $(\agentsteam{k} + 1)^{|\StSpteam{k}|}$ elements, which increase polynomially with the number $\agentsteam{k}$ of agents. Another added challenge is that it is complicated to explicitly construct the conditional distribution $Q^{(k)}$ used in property (P2) of Proposition~\ref{prop:IS}. An approach to bypass both difficulties is to use sampling-based reinforcement learning techniques~\cite{chang2007simulation} which do not explicitly construct the action-value function. Sampling based techniques are particularly efficient in our setting because the condition distribution $Q^{(k)}$ are defined via a series of products of multinomial distributions \eqref{eq:m-to-bar-m} and \eqref{eq:bar-m-to-hat-m}, so we can efficiently sample from $Q^{(k)}$ without constructing it explicitly.
\end{remark}

\begin{remark}
  Sampling from the multinomial distribution is still computationally involved when the number $\agentsteam{(k)}$ of agents in each team is large. In such a setting, one option is to approximate the dynamics $Q^{(k)}$ by using the multivariate Gaussian approximation to the multinomial distribution. 
\end{remark}
\section{Mean-field approximation}\label{sec:mf}

In this section, we consider the setting in which each team has a large number of
players. We approximate the system by an infinite population mean field limit
and show that any MPE of the infinite population game is an
$\varepsilon$-Team-Nash equilibrium of the original game, where $\varepsilon =
(\varepsilon_{k})_{k \in \teamset}$ and $\varepsilon_k =
\mathcal{O}(1/\sqrt{\agentsteam{k}})$. 

The infinite population approximation provides a drastic simplification of the
dynamic program of Theorem~\ref{thm:MPE} because in the infinite population
mean-field system, the mean-field is equivalent to the statistical
distribution of the agents and, therefore, evolves in a deterministic manner. 
\subsection{Regularity conditions on the system} 

When there are $\agentsteam{k}$ agents in team~$k$, the mean-field
$\MFteam{k}_t$ takes values in $\MFteamSpace{k}$, which is the set of all
probability mass functions (PMFs) $\mfteam{k}$ on $\StSpteam{k}$ such that
$\agentsteam{k} \mfteam{k}$ is a vector of non-negative integers. Note that
$\MFteamSpace{k} \subset \Delta(\StSpteam{k})$, the set of all PMFs on
$\StSpteam{k}$. To understand the behavior of the finite population system as
the number of agents become large, we first extend the domain of $\mfteam{k}$
in the cost function and the transition functions from $\MFteamSpace{k}$ to
$\Delta(\StSpteam{k})$. In particular, we let $\SMFSpace$ denote the set
$\prod_{k \in \teamset} \Delta(\StSpteam{k})$ and assume that the per-step
cost function $\costteamagent{k}_t$ is a function from $\StSpteam{k} \times
\AcSpteam{k} \times \SMFSpace$ to $\reals$ and the controlled transition
matrix $\stprobteam{k}$ is a transition matrix from $\StSpteam{k} \times
\AcSpteam{k} \times \SMFSpace$ to $\StSpteam{k}$. 

We start with some formal definitions needed to define the Lipschitz continuity of
$\costteamagent{k}_t$ and $\stprobteam{k}$. 
Let $d^{(k)}_{\st}$ be a metric on the state space $\StSpteam{k}$, $k \in \teamset$. Then, based on this metric, let $d_{\mathfrak{w}}^{(k)}$ be the Kantorovich metric (also called Wasserstein metric) on $\Delta(\StSpteam{k})$ (i.e., the space of mean-fields for each team) $k \in \teamset$. Define a metric $\W$ on the set of mean-fields for all teams $\SMFSpace$ as:
\begin{equation} \label{eq:z-metric}
    \W(\mf, \hat{\mf}) = \sum_{k \in \teamset}\Bigl(d_{\mathfrak{w}}^{(k)}(\mfteam{k}, \hat \mf^{(k)})\Bigr),
    \quad
    \mf, \hat \mf \in \SMFSpace.
\end{equation}
\subsection{Infinite population mean-field approximation}

We now consider an infinite population approximation of Game~\ref{game:common}, where we approximate the per-step cost $\costfunctionteam{k}_t$ and the dynamics $Q^{(k)}$, defined in Proposition~\ref{prop:IS} by $\scostfunctionteam{k}$ and $\bar Q^{(k)}$ defined as follows:
\begin{align}
  \scostfunctionteam{k}_t(\smf, \prescriptionteam{k}) &= \sum_{\st \in \StSpteam{k}}\smfteam{k}(\st)\costteamagent{k}_t(\st, \prescriptionteam{k}(\st), \smf), \label{eq:smf-cost-team}
  \\
  \bar Q^{(k)}(\smfteam{k} \mid \smf, \prescriptionteam{k}) &=
  \IND\{ \smfteam{k} = \bar q^{(k)}(\smf, \prescriptionteam{k)})\}, \label{eq:smf-transition-team}
\shortintertext{where}
  \bar q^{(k)}(\smf, \prescriptionteam{k)}) &=
  \sum_{\st \in \StSpteam{k}}\smfteam{k}_t(\st)P^{(k)}(\st'|\st, \prescriptionteam{k}_t(\st), \smf_t),
\label{eq:mf-evolve-inf}
\end{align}
and under the mean-field dynamics
\begin{equation}\label{eq:mf-evolve-inf-total}
\PR(\MF_{t+1}|\MF_t = \mf_t, \Prescription_t = \prescription_t) \eqqcolon \bar q(\smf_t, \prescription_t) = 
  \prod_{k \in \teamset}\bar q^{(k)}(\smf, \prescriptionteam{k)}).
\end{equation}
\begin{lemma} \label{lem:ais-inf}
   The infinite population game is an approximation of the finite population game in the following sense. 
   \begin{enumerate}
       \item For any $k \in \teamset$, we have
       \[
          \bigl| \costfunctionteam{k}_t(\mf, \prescriptionteam{k}) - \scostfunctionteam{k}_t(\mf, \prescriptionteam{k})  \bigr| = 0, \forall \mf \in \MFSpace, \prescriptionteam{k}
       \]
       \item For any $k \in \teamset$, we have
       \[
          \W(\bar q(\mf_t, \prescription_t), Q(\cdot \mid \mf_t, \prescription_t)) {\le}\sum_{k \in \teamset}\frac{\kappa}{\sqrt{\agentsteam{k}}}  =: \delta_t
        \]
        where $\kappa$ is a constant that depends on the state spaces $\StSpteam{k}$ and the metric $d_\st$.
   \end{enumerate}
\end{lemma}
\begin{proof}
The first part of the lemma follows from the definitions of $\costfunctionteam{k}$ and $\scostfunctionteam{k}$. For the second point, we first note that from 
~\cite[Lemma 4]{sinha2023sensitivity} we have:
\begin{equation*}
    \W(\bar q(\mf_t, \prescription_t), Q(\cdot \mid \mf_t, \prescription_t)) \le \sum_{k \in \teamset} \W(Q^{(k)}(\cdot \mid \mf_t, \prescriptionteam{k}_t), \bar q^{(k)}(\mf_t, \prescriptionteam{k}_t),
\end{equation*}
Furthermore,  concentration of empirical measure to statistical measure with respect to the Wasserstein distance ~\citep{Sommerfeld:2019} implies that
\begin{equation*}
    \W(\bar q^{(k)}(\mf_t, \prescriptionteam{k}_t, Q^{(k)}(\cdot \mid \mf_t, \prescriptionteam{k}_t)) \le \frac{\kappa}{\sqrt{\agentsteam{k}}},
\end{equation*}
where $\kappa$ is a constant that depends on the state spaces $\StSpteam{k}$ and the metric $d_\st$. Combining the above two equations implies the second result.
\end{proof}
Since the infinite population approximation is a Markov game, its
Markov perfect equilibrium is characterized as follows~\citep{maskin1988theory,maskin1988theory2}.
\begin{theorem} \label{thm:MPE-inf}
  Consider a strategy profile $\scorlaw = (\scorlawteam{k})_{k \in \teamset}$,
  where each virtual player is playing a Markov strategy. 

  A necessary and sufficient condition for $\scorlaw$ to be a Markov perfect
  equilibrium for the mean-field limit of Game~\ref{game:common} is that it satisfy the following conditions:
  \begin{enumerate}
    \item For each possible realization $\smf_T$ of $\SMF_T$, define the value
      function for virtual player~$k$:
      \begin{equation} \label{eq:DP:T-inf}
        \widebar{V}^{(k)}_T(\smf_T) = \min_{\prescriptionteam{k}_T} 
        \scostfunctionteam{k}_T(\smf_T, \prescriptionteam{k}_T).
      \end{equation}
      Then, $\scorlawteam{k}_T(\smf_T)$ must be a minimizing
      $\prescriptionteam{k}_T$ in~\eqref{eq:DP:T-inf}.

    \item For $t \in \{T-1, \dots, 1\}$ and for each possible realization
      $\smf_t$ of $\SMF_t$, recursively define the value function for virtual
      player~$k$:
      \begin{align} \label{eq:DP:t-inf}
        \widebar{V}^{(k)}_t(\smf_t) &= \min_{\prescriptionteam{k}_t}
        \Big\{ \scostfunctionteam{k}_t(\smf_t, \prescriptionteam{k}_t)  
            + \widebar V^{(k)}_{t+1} \bigl( (q^{(k)}(\smf_t, \prescriptionteam{k}_t) )_{k \in \mathcal K} \bigr)
         \Bigr\}.
      \end{align}
      Then $\scorlawteam{k}_t(\smf_t)$ must
      be a minimizing $\prescriptionteam{k}_t$ in~\eqref{eq:DP:t-inf}.
  \end{enumerate}
\end{theorem}
\subsection{$\epsilon$-Team-Nash equilibrium}
Now we address the fundamental question of the mean-field approximation: is the infinite poulation approximation good and, if so, in what sense? We show that the any MPE of the infinite population limit is an approximate MPE of the original finite population game, where the approximation error scales as $\mathcal{O}(1/N)$.

\begin{theorem}\label{thm:approxMPEInf}
    Suppose there exists an MPE $\bar \psi$ for the infinite population game such that the corresponding total costs $\widebar L^{(k)}_t$ are $\mathcal{L}^{(k)}_t$-Lipschitz, $k \in \teamset$, $t \in \text{time}$. Then, $\bar \psi$ is an $\mathcal{O}(1/\sqrt{N})$-approximate MPE of the finite population game, where $N = \inf N^{(k)}$. In particular, for any other strategy $\corlawteam{k}$ of team~$k$, we have
    \[
        \widebar L^{(k)}(\scorlawteam{k}, \scorlawteam{-k}) \le \widebar L^{(k)}(\corlawteam{k}, \scorlawteam{-k}) + 2\sum_{t=1}^T \sum_{k \in \teamset} \frac{\kappa^{(k)} \mathcal{L}^{(k)}_t}{\sqrt{N^k}} 
    \]
\end{theorem}
\begin{proof}
See Appendix~\ref{app:proof2}
\end{proof}

\section{Conclusion}\label{sec:conclusion}
In this paper, we presented a model for mean-field games among teams and presented a common-information based refinement of the Team-Nash equilibrium for this game. This common-information based Markov perfect equilibrium can be obtained by solving coupled dynamic programs. These dynamic programs use the mean field of all teams as a state. In general, solving such dynamic programs suffers from the curse of dimensionality. To circumvent this curse of dimensionality, we use  a mean-field limit to approximate finite population teams by an infinite population. We show that a Markov perfect equilibrium obtained using the mean-field approximations is an approximate Markov perfect equilibrium of the orignal game.

\bibliographystyle{abbrvnat}
\bibliography{IEEEabrv,mfgamesteams}

\appendix

\section{Proof of Theorem~\ref{thm:equiv}}\label{app:proof1}
We first establish some intermediate results.
\begin{lemma} \label{lem:pf-1-lem-1}
  Given a strategy $\pol$ of Game~\ref{game:original}, let $\corlawhist$ be a strategy of Game~\ref{game:common} constructed according to~\eqref{eq:NE-forward}. Then, for any time $t$ and any realizations $\mf_{1:t}$ and $\prescription_{1:t}$ of $\MF_{1:t}$ and $\Prescription_{1:t}$, we have that
  \[
    \PR^{\pol}(\MF_{1:T} = \mf_{1:T}, \Prescription_{1:T} = \prescription_{1:t})
    = 
    \PR^{\corlawhist}(\MF_{1:T} = \mf_{1:T}, \Prescription_{1:T} = \prescription_{1:t})
 \]
\end{lemma}
\begin{proof}
    For simplicity of notation, we write $\PR^{\pol}(\mf_{1:t}, \prescription_{1:t})$ instead of 
    $\PR^{\pol}(\MF_{1:T} = \mf_{1:t}, \Prescription_{1:T} = \prescription_{1:t})$ and use similar shortcuts for other terms as well.
    
    The proof proceeds by induction on $T$. For $t=1$, both sides of the equation do not depend on the strategies, and the result is trivially true. This forms the basis of induction. Now, suppose that the result is true for $t$ and consider the system at $t+1$. From property (P2) of Proposition~\ref{prop:IS}, we have
    \begin{equation}\label{pf:1-s-1}
        \PR^{\pol}(\mf_{t+1} \mid \mf_{1:t}, \prescription_{1:t})
        = Q( \mf_{t+1} \mid \mf_t, \prescription_t )
    \end{equation}
    and similarly
    \begin{equation}\label{pf:1-s-2}
        \PR^{\corlawhist}(\mf_{t+1} \mid \mf_{1:t}, \prescription_{1:t})
        = Q( \mf_{t+1} \mid \mf_t, \prescription_t ).
    \end{equation}
    Furthermore, the construction of strategy $\corlawhist$ implies that 
    \begin{equation}\label{pf:1-s-3}
        \PR^{\pol}(\prescription_{t+1} \mid \mf_{1:t+1}, \prescription_{1:t})
        =
        \PR^{\corlawhist}(\prescription_{t+1} \mid \mf_{1:t+1}, \prescription_{1:t}).
    \end{equation}
    Properties~\eqref{pf:1-s-1}--\eqref{pf:1-s-3} along with the induction hypothesis implies that
    \begin{align}
        \PR^{\pol}(\mf_{1:t+1}, \prescription_{1:t+1}) &=
        \PR^{\pol}(\prescription_{t+1} \mid \mf_{1:t+1}, \prescription_{1:t})
        \PR^{\pol}(\mf_{t+1} \mid \mf_{1:t}, \prescription_{1:t})
        \PR^{\pol}(\mf_{1:t}, \prescription_{1:t})
        \notag \\
        &=
        \PR^{\corlawhist}(\prescription_{t+1} \mid \mf_{1:t+1}, \prescription_{1:t})
        \PR^{\corlawhist}(\mf_{t+1} \mid \mf_{1:t}, \prescription_{1:t})
        \PR^{\corlawhist}(\mf_{1:t}, \prescription_{1:t})
        \notag \\
        &= \PR^{\corlawhist}(\mf_{1:t+1}, \prescription_{1:t+1}).
    \end{align}
    This completes the induction step.
\end{proof}
  
\begin{lemma} \label{lem:pf-1-lem-2}
  Given a strategy $\corlawhist$ of Game~\ref{game:common}, let $\pol$ be a strategy of Game~\ref{game:original} constructed according to~\eqref{eq:NE-reverse}. Then, for any time $t$ and any realizations $\mf_{1:t}$ and $\prescription_{1:t}$ of $\MF_{1:t}$ and $\Prescription_{1:t}$, we have that
  \[
    \PR^{\corlawhist}(\MF_{1:T} = \mf_{1:T}, \Prescription_{1:T} = \prescription_{1:t})
    = 
    \PR^{\pol}(\MF_{1:T} = \mf_{1:T}, \Prescription_{1:T} = \prescription_{1:t})
 \]
\end{lemma}
The proof is similar to the proof of Lemma~\ref{lem:pf-1-lem-1} and is omitted.  

\begin{lemma}\label{lem:pf-1-lem-3}
    Given a strategy $\pol$ of Game~\ref{game:original}, let $\corlawhist$ be a strategy of Game~\ref{game:common} constructed according to~\eqref{eq:NE-forward}. Then,
    \[
        \perfteam{k}(\pol) = \Costfunctionteam{k}(\corlawhist).
    \]
\end{lemma}
\begin{proof}
  Arbitrarily fix a team~$k \in \teamset$ and consider
  \begin{equation}
      \EXP^{\pol}[ \costteam{k} ] \stackrel{(a)}= \EXP^{\pol}\bigl[ \EXP[ \costteam{k} \mid \MF_{1:t}, \prescription_{1:t}] \bigr] 
      \stackrel{(b)}= \EXP^{\pol}[ \costfunctionteam{k}(\MF_t,\prescriptionteam{k}_t) ]
      \label{eq:pf-1-step-1}
 \end{equation}
 where $(a)$ follows from the smoothing property of conditional expectation and $(b)$ follows from (P1) in Proposition~\ref{prop:IS}. Therefore,
 \begin{align}
    \perfteam{k}(\pol) &\stackrel{(c)}= \EXP^{\pol}\biggl[ \sum_{t=1}^T \costfunctionteam{k}(\MF_t,\prescriptionteam{k}_t) \biggr] 
    \notag \\
   &\stackrel{(d)}= \EXP^{\corlawhist}\biggl[ \sum_{t=1}^T \costfunctionteam{k}(\MF_t,\prescriptionteam{k}_t) \biggr]  
   \notag \\
   &= \Costfunctionteam{k}(\corlawhist)
 \end{align}
where $(c)$ follows from~\eqref{eq:pf-1-step-1} and $(d)$ follows from Lemma~\ref{lem:pf-1-lem-1}.
\end{proof}

\begin{lemma} \label{lem:pf-1-lem-4}
  Given a strategy $\corlawhist$ of Game~\ref{game:common}, let $\pol$ be a strategy of Game~\ref{game:original} constructed according to~\eqref{eq:NE-reverse}.  Then,
    \[
        \perfteam{k}(\pol) = \Costfunctionteam{k}(\corlawhist).
    \]
\end{lemma}
The proof argument is similar to that of Lemma~\ref{lem:pf-1-lem-3} and is omitted.

Now we are ready to prove the result of Theorem~\ref{thm:equiv}. Let $\pol$ be a NE for Game~\ref{game:original} and let $\corlawhist$ be a strategy for Game~\ref{game:common} constructed according to~\eqref{eq:NE-forward}. Suppose $\corlawhist$ is not a NE for Game~\ref{game:common}. That is, there exists a team $k \in \teamset$ and a strategy $\corlawteam{k}$ for team~$k$ such that
\[
    \Costfunctionteam{k}(\corlawteam{k}, \corlawhistteam{-k}) 
    <
    \Costfunctionteam{k}(\corlawhistteam{k}, \corlawhistteam{-k}).
\]
Let $\bar \pi^{(k)}$ be the strategy for Game~\ref{game:original} corresponding to $\corlawteam{k}$ constructed according to~\eqref{eq:NE-reverse}. Then, Lemmas~\ref{lem:pf-1-lem-4} implies that
\[
   \perfteam{k}(\bar \pi^{(k)}, \polteam{-k}) 
   <
   \perfteam{k}(\polteam{k}, \polteam{-k}) 
\]
contradicting the fact that $\pol$ is a NE for Game~\ref{game:original}. Therefore, $\corlawhist$ must be a NE of Game~\ref{game:common}.

The second part of the theorem can be proved by an analogous argument.

\section{Proof of Theorem~\ref{thm:approxMPEInf}}\label{app:proof2}
\begin{proof}
Fix a virtual player $k$ and the strategy profile $\scorlawteam{-k}$ for the other virtual players and consider the best response dynamics at virtual player $k$ given by the dynamic program in Thm.~\ref{thm:MPE}. The idea of the proof is to show that the history compression function $\nu_t(\mf_{1:t}, \prescription_{1:t}) = \mf_t$ dynamics $(q^{(k)}_t)_{k \in \teamset}$ and the per-step cost $\scostfunctionteam{k}_t$ is an approximate information state (AIS) as defined in~\cite{aisjmlr}. In particular, we observe that:
\begin{align} \label{eq:ais-1}
    \EXP[\costfunctionteam{k}_t(\mf, \prescriptionteam{k}) - \scostfunctionteam{k}_t(\nu_t(\mf_{1:t}, \prescription_{1:t)}, \prescriptionteam{k})] = 0 \coloneqq \varepsilon_t,
\end{align}
and
\begin{align} \label{eq:ais-2}
    d_{\teamset}&(\PR(\MF_{t+1}|\MF_t = \mf_t, \Prescription_t = \prescription_t), q_t(\mf_t, \prescription_t)) \le\sum_{k \in \teamset}\frac{\kappa}{\sqrt{\agentsteam{k}}} \coloneqq \delta_t,
\end{align}
which follow from Lemma~\ref{lem:ais-inf}. Equations~\eqref{eq:ais-1}, \eqref{eq:ais-2} show that $(\nu_t, (q^{(k)}_t)_{k \in \teamset}, \scostfunctionteam{k}_t)$ is an AIS. Then, the result follows from~\cite[Theorem 9]{aisjmlr}.
\end{proof}

\end{document}